\documentclass[letterpaper, 10 pt, journal, twoside]{IEEEtran}
\usepackage{blindtext, graphicx}
\usepackage{color}
\usepackage{amsmath, amssymb}
\usepackage{siunitx}
\usepackage{dblfloatfix}
\usepackage{xcolor}

\usepackage{amsthm}

\newtheorem{theorem}{Theorem}
\newtheorem{corollary}{Corollary}
\newtheorem{lemma}{Lemma}
\newtheorem{remark}{Remark}

\title{Transient stability guarantees for ad hoc dc microgrids}
\author{Kathleen Cavanagh, Julia A. Belk,  Konstantin Turitsyn \thanks{KC (kcav@mit.edu) and KT are with the  Mechanical Engineering Department, Massachusetts Institute of Technology, Cambridge, MA. JAB is with the Computer Science Department, Stanford University, Stanford, CA. This work was sponsored by the MIT Tata Center for Technology and Design}.
}

\begin{document}

\maketitle
\thispagestyle{empty}
\pagestyle{empty}

\begin{abstract}
Ad hoc electrical networks are formed by connecting power sources and loads without planning the interconnection structure (topology) in advance. They are designed to be installed and operated by individual communities---without central oversight---and as a result are well-suited to addressing the lack of electricity access in rural and developing areas. However, ad hoc networks are not widely used, and a major technical challenge impeding their development (and deployment) is the difficulty of certifying network stability without a priori knowledge of the topology. We develop conditions on individual power sources and loads such that a microgrid comprised of many units will be stable. We use Brayton-Moser potential theory to develop design constraints on individual microgrid components that certify transient stability---guaranteeing that the system will return to a suitable equilibrium after load switching events. Our central result is that stability can be ensured by installing a parallel capacitor at each constant power load, and we derive an expression for the required capacitance.
\end{abstract}

\begin{IEEEkeywords}
Lyapunov methods, Network analysis and control, Power systems, Stability of nonlinear systems
\end{IEEEkeywords}

\section{Introduction}



\IEEEPARstart{M}{icrogrids} are smaller-scale than conventional power systems and can be designed to naturally incorporate distributed renewable resources. These benefits make microgrids attractive for addressing the lack of electricity in remote and rural areas, which continues to affect more than one billion people \cite{iea15}. The need for universal electricity access and evolving demands on existing bulk power infrastructure have driven extensive development of microgrids in recent years, but the capital-intensive planning process and the need for centralized control continue to impede adoption.





\textit{Ad hoc microgrids}---microgrids that can be set up without predetermining the network structure---reduce the financial barrier to energy access. Instead, they pose a technical challenge: network stability must be certified before the network topology is known, and the network may be modified after installation depending on the community's needs. We focus on a previously-presented low voltage dc architecture designed specifically for rural electrification \cite{ugrid}.

Like all power systems, microgrids are not globally stable, which presents unique control challenges. Further, there are three unusual features of our analysis that require a significant departure from traditional power systems methods: (1) the ad hoc setting, (2) the use of tightly-regulated power electronics at all sources and loads, and (3) transient stability guarantees. Using power electronics to interface sources and loads to the network offer new opportunities for decentralized and autonomous control of power supply and demand \cite{cdc2016}, but they complicate system stability because they draw constant power from the network to regulate their outputs. The negative incremental resistance ($\partial v/\partial i$) of these loads has a destabilizing effect on power systems \cite{cplinstability}. The impact of constant power loads on the stability of conventional microgrids has attracted interest in the controls community \cite{Sanchez:2013gl,Bolognani:2015ek,SimpsonPorco:2015hp,Cezar:2015io,Barabanov:2016ki} and the power electronics community \cite{CriteriaReview,ROM}. Previous analyses are based on simplified models of resistive lines that remain stable under mild constraints and arbitrarily high control gains. However, in practical settings the line inductance is a source of instability and cannot be neglected.

To our knowledge, all previous studies have focused on known and fixed topologies, impractically simplified models, and/or linearized models. Our analysis encompasses unknown and changing topologies, realistic subsystem models, and significantly extends our previous conference papers \cite{cdc2016,compel} by providing transient stability guarantees for our nonlinear architecture. Robustness to large perturbations is important, especially in low-voltage, low-power networks where each household may be a significant fraction of the total system load. Mathematically, certifying transient stability requires characterizing the extrema \textit{and} attraction basins of our previously presented potential function for networks of unknown topologies. Our main contribution is a set of design-friendly constraints on individual network units (sources, loads, and lines), summarized by Eqs. \eqref{eq:CVtr}, \eqref{eq:Cbound} and \eqref{eq:Css}.

\section{Models and Notation}\label{sec:models}

In this section we present models for the interconnecting lines, power electronic loads, and voltage source converters which are analytically tractable and can be adapted to describe many networks. These are based on a previously presented ad hoc microgrid \cite{ugrid}.


The electrical network is described as a weighted, directed graph $({\cal V,\cal E})$ with a total of  $|{\cal V}|$ nodes (buses) and $|{\cal E}|$ edges (lines). Each edge $\alpha\in\mathcal{E}$ represents a tuple $\alpha=(i,j)$ with $i,j\in\mathcal{V}$. A power source or load is attached to each node and we denote the subset of vertex indices corresponding to loads as $\mathcal{L} \subset {\cal V}$ with and the subset of source indices as $\mathcal{S} \subset {\cal V}$. The state of the system is described by the voltage and current vectors $v \in \mathbb{R}^{|\mathcal{V}|}$ and $i \in \mathbb{R}^{|\mathcal{E}|}$. The topology of the graph is defined by the (transposed) incidence matrix $\nabla \in \mathbb{R}^{|\mathcal{E}|\times |\mathcal{V}|}$. Applying $\nabla$ to the voltage vector results in a vector of potential differences across each line, and applying $\nabla^\top$ to the current vector yields the total current flowing out of each node.





Each power line is associated with a graph edge $\alpha = (i,j)\in\mathcal{E}$ and is characterized by an inductance $L_\alpha = L_{ij}$ and a resistance $R_\alpha = R_{ij}$. Each line has time constant $\tau_\alpha = L_\alpha/R_\alpha$. Each line current $i_{\alpha}$ is described by:
\begin{equation}\label{eq:Current_Dynamics}
    L_\alpha \dot i_\alpha = - R_\alpha i_\alpha + \sum_{k\in\mathcal{V}} \nabla_{\alpha k} v_k, \quad \alpha \in {\cal E}\,.
\end{equation}

Load $k$ is represented by the parallel connection of a capacitance $C_k$ and a constant power load drawing power $p_k$. In general, constant power loads represent perfectly-regulated power converters with constant resistance loads, and hence are conservative and general models which can be used to describe many power electronic devices. The dynamics of these converters are much faster than the inductive response of the lines. The capacitor across the input of the power converter is a standard feature, and is critical for system stability \cite{Cezar:2015io}. Each load voltage is described by:

\begin{equation}\label{eq:Voltage_Dynamics}
    C_k \dot{v}_k = -\frac{p_k}{v_k} - \sum_{\alpha\in \mathcal{E}} \nabla_{\alpha k} i_\alpha, \quad k \in \mathcal{L}\,.
\end{equation}

In this work we assume the source controller, which regulates the converter output voltage, has dynamics much faster than the network. Accordingly, we model them as perfect voltage sources: $v_k = V_0,\quad k \in \mathcal{S}$. The extension to controllable converters with proportional (droop) and integral voltage control is relatively straightforward \cite{cdc2016}.

Characterizing system stability requires a suitable family of Lyapunov (potential) functions. Extrema of a particular potential correspond to equilibria of the system, and stability can be certified by demonstrating certain properties of the potential. Unfortunately, the presence of constant power loads and lack of global stability in power systems preclude the use of potential functions based on system energy (Hamiltonian potentials). However, the seminal results of Brayton and Moser are applicable to our setting \cite{Brayton:1964gr, Jeltsema:2009jd}. The Brayton Moser potential represents the system dynamics in a quasi-gradient form ($\mathcal{Q} \dot{x} = -\partial_x \mathcal{P}$), which is particularly useful for certification of transient stability. Whenever the matrix $\mathcal{Q}$ is positive definite in some region, the potential $\mathcal{P}$ is a non-increasing function: $\dot{\mathcal{P}} =  - \dot{x}^T \mathcal{Q}\dot{x}$. Lower and upper bounds on $\mathcal{P}$ that hold for arbitrary networks can be used to establish stability for ad hoc networks.



For our architecture, a representation with the proper structure (derived in \cite{cdc2016}) is given by:


\begin{align}
    \mathcal{G}(v) ={}& \sum_{(i,j)\in\mathcal{E}} \frac{(v_i -v_j)^2}{2 R_{ij}} + \sum_{k\in\mathcal{L}}p_k\log v_k \label{eq:Gdef} \\
    \begin{split}\label{Pdef}
    \mathcal{P}(x) ={}& \mathcal{G} (v) + \frac{1}{2} \sum_{\alpha \in\mathcal{E}} (\tau_{\mathrm{max}} -\tau_\alpha)L_\alpha\dot{i}_\alpha^2\\
    &\quad \quad + \frac{\tau_{\mathrm{max}}}{2}\sum_{k\in\mathcal{L}} C_k\dot{v}^2_k
    \end{split}
\end{align}
\begin{align}
    \label{eq:Qdef}
 \mathcal{Q} &= \begin{bmatrix}
\mathrm{diag}\left(\tau_{\mathrm{max}} R_\alpha - L_\alpha\right) & -\tau_{\mathrm{max}}\nabla_{\mathcal{E}\mathcal{L}} \\
\tau_{\mathrm{max}} \nabla_{\mathcal{E}\mathcal{L}}^T & \mathrm{diag}\left(C_k - \tau_{\mathrm{max}} \frac{p_k}{v_k^2}\right)
\end{bmatrix}
\end{align}
$x$ is the state vector $[i_\mathcal{E}^T, v_\mathcal{L}^T]^T$ and $\dot{i}_\alpha$, $\dot{v}_k$ are given in Eqs. \eqref{eq:Current_Dynamics} and \eqref{eq:Voltage_Dynamics}.  $\tau_{\max}$ is an upper bound on the line time constants in the network---for convenience, we define it to be strictly larger than the largest time constant: $\tau_{\max} > \max_\alpha \tau_\alpha$. $\nabla_{\mathcal{E}\mathcal{L}}$ refers to the submatrix of $\nabla$ corresponding to the load nodes. Finally, in addition to being notationally convenient, ${\cal G}$ is well-studied and is referred to in the literature as the resistive co-content \cite{Jeltsema:2009jd}. It is also worth noting, that all equilibria of the system correspond to extrema of $\mathcal{P}$, and vice-versa, every extremum of $\mathcal{P}$ to an equilibrium of the system. Moreover, the extrema of the potential $\mathcal{G}(v)$ correspond to the solutions of the equilibrium power flow equations.

\section{Stability of a Two Bus System}\label{sec:twobus}

In this section we consider the two bus system shown in Figure \ref{fig:LineDiag_2Bus}, to provide a simple introduction to the techniques we will use in the next section to analyze ad hoc networks with no topology constraints. 
The dynamic equations of the system are given by:
\begin{align}
C \dot{v} &=-\frac{p}{v}+i\\
L \dot{i} &=-Ri+V_0-v\,.
\end{align}


\begin{figure}[b!]
    \centering
    \includegraphics[width=0.8\linewidth]{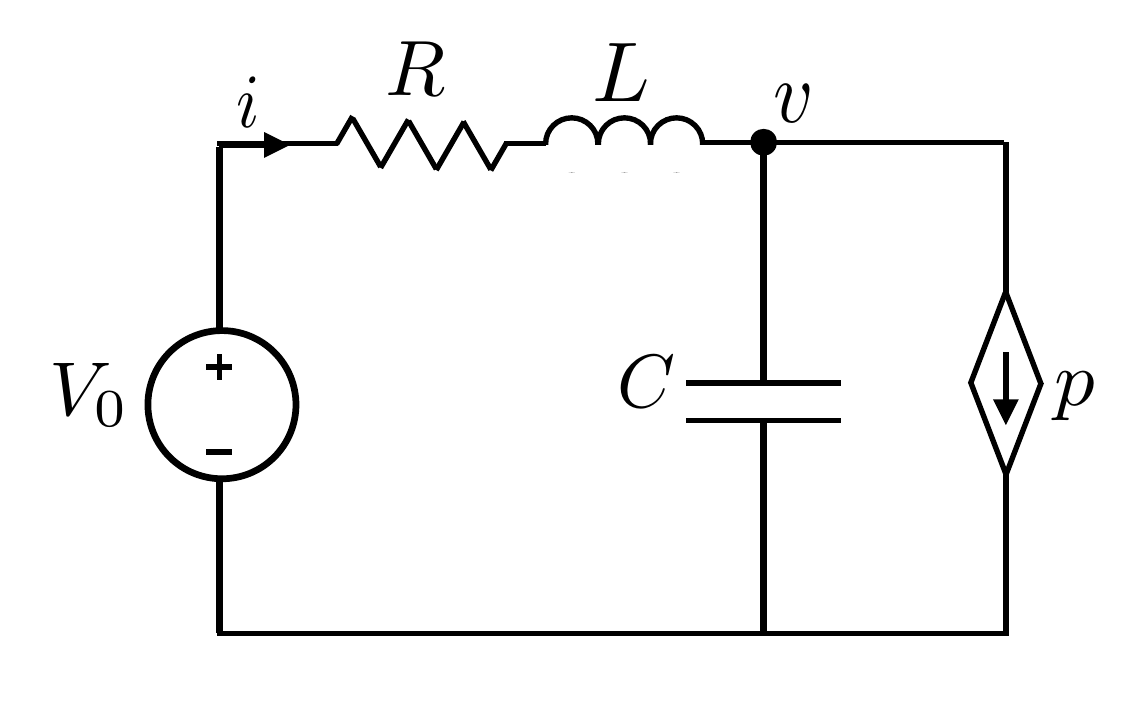}
    \caption{Description of two bus network}
    \label{fig:LineDiag_2Bus}
\end{figure}

The relationship between the voltage $v$ at the load and the load magnitude $p$ at equilibrium (the ``nose curve'') is shown in Figure \ref{fig:NoseCurve}. The largest load that can be supported is the apex of the nose curve at $p = P_0 = V_0^2/4R$, which corresponds to a load bus voltage of $V_0/2$. For all $p < P_0$, there are two solutions:

\begin{align}\label{eq:Vhigh}
    V_\mathrm{high}(p, R) &=\tfrac{V_0}{2} \left(1+\sqrt{1-\tfrac{p}{P_0(R)}}\right)\\\label{eq:Vlow} V_\mathrm{low}(p, R) &= \tfrac{V_0}{2}\left(1-\sqrt{1-\tfrac{p}{P_0(R)}}\right)
\end{align}
where $V_\mathrm{high}$ and $V_\mathrm{low}$ are the high and low equilibrium points. When the power exceeds $P_0$, $\mathcal{G}$ does not have any extrema and the system has no equilibria. For the network to achieve a minimum voltage of $V_0 > V_{\min} > V_0/2$, the largest load that can be supported is $P_{\max} = V_{\min}(V_0-V_{\min})/R$.

\begin{figure}
    \includegraphics[width = 0.8\linewidth]{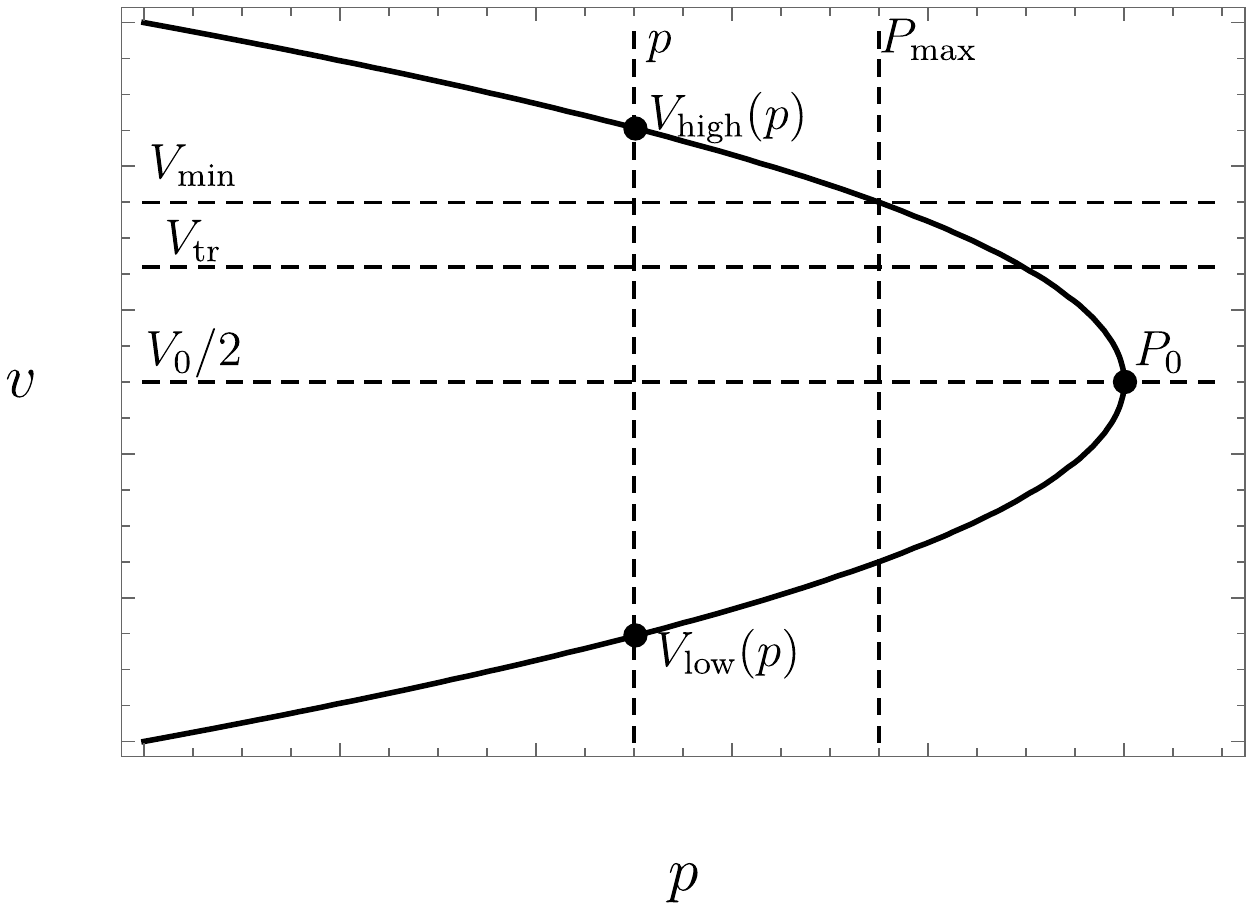}
    \caption{Nose curve demonstrating the relationship between power and load bus voltage}
    \label{fig:NoseCurve}
\end{figure}





To analyze transient stability, we first define a ``switching event'' to be any time $t$ such that the load power changes instantaneously:
\begin{align}
    p(t^-) \neq p(t^+)\,.
\end{align}
Hereafter, $^+$ and $^-$ refer to quantities evaluated at $t^+$ and $t^-$. All state variables are continuous: $i^-=i^+$ and $v^- = v^+$. We assume that the system is at equilibrium before the switching event: $v^- = V_\mathrm{high}^-$.

Applying the definitions of $\mathcal{G}$ and $\mathcal{P}$ from Sec. \ref{sec:models} yields:
\begin{align}
\label{eq:G2Bus}
\mathcal{G}(v)&=\frac{(V_0-v)^2}{2 R} + p\log v\\
\label{eq:P2Bus}
\mathcal{P}(x)&=\mathcal{G}+ \frac{1}{2}(\tau_{\mathrm{max}} - \tau)L \dot{i}^2+\frac{\tau_{\mathrm{max}}}{2} C \dot{v}^2\,.
\end{align}
According to Lyapunov's theorem, convergence to a stable equilibrium point after a switching event is guaranteed whenever 1) the potential $\mathcal{P}$ is strictly decreasing with time:
\begin{equation}\label{eq:Pdot}
\dot{\mathcal{P}} = - R(\tau_{\mathrm{max}}-\tau)\dot{i}^2 -\left(C-\frac{\tau_{\mathrm{max}}p}{v^2}\right)\dot{v}^2 <0,
\end{equation}
and 2) the high voltage domain is invariant and contains a single equilibrium point
 and 3) the post switching potential evaluated at the pre-switching equilibrium, due to the continuity of state variables, is less than the potential of the bound on the voltage domain: $\mathcal{P}^+(V_{\mathrm{high}}^-) <   \mathcal{P}^+(V_\mathrm{tr})$.

To satisfy the first requirement, we note that the first term of Equation \eqref{eq:Pdot} is negative by definition. Therefore, to ensure the second term is also negative, a lower bound on the load bus voltage during transients, denoted by $V_\mathrm{tr}$, is imposed such that $v\geq V_\mathrm{tr}$. The choice of $V_\mathrm{tr}$ provides a bound on $C$ such that
\begin{equation}
\label{eq:VtrBound}
    C>\frac{\tau P^{\max}}{V_\mathrm{tr}^2}
\end{equation}

To satisfy the second requirement, we must choose $V_\mathrm{tr}$ such that the stable equilibrium point is in the domain and the unstable point, in a two bus system $V_\mathrm{low}$, is not. Therefore the value of $V_\mathrm{tr}$ is bounded from above by $V_{\min}$, the lowest acceptable equilibrium point, and from below by $V_\mathrm{low}$. For simplicity, we impose a stricter lower bound of $V_0/2$. $V_\mathrm{tr}$ is thus bounded such that $V_{\min} \geq V_\mathrm{tr} \geq V_0/2$.

The third requirement can be simplified given that $\mathcal{P}^+\left(V_\mathrm{tr}\right)$ is bounded from below by $\mathcal{G}^+\left(V_\mathrm{tr}\right)$ such that
\begin{equation}\label{eq:transient}
\mathcal{P}^+(V_{\mathrm{high}}^-) <   \mathcal{G}^+(V_\mathrm{tr})
\end{equation}
The transient component of $\mathcal{P}$ can be expressed in terms of the state variables by characterizing the capacitor current as $C\dot{v}(t^+)=(p^--p^+)/V^-_\mathrm{high}$.
which then yields
\begin{equation}
\mathcal{P}^+=\mathcal{G}^+(V_\mathrm{high}^-)+\frac{\tau}{2C}\left ( \frac{p^--p^+}{V^-_\mathrm{high}} \right )^2\,.
\end{equation}
This representation of $\mathcal{P}$ makes explicit that larger values of $C$ decrease the total potential. Therefore, the third stability condition yields a bound $C > C_{\mathrm{tr}}(p^-,p^+)$ where $C_{\mathrm{tr}}$ is given by:
\begin{equation}\label{eq:Climit}
 C_{\mathrm{tr}}(p^-,p^+ ) = \frac{\tau  }{2(\mathcal{G}^+(V_\mathrm{tr})-\mathcal{G}^+(V_\mathrm{high}^-))}\left(\frac{p^--p^+}{V_\mathrm{high}^-}\right)^2\nonumber
\end{equation}
This expression can be reduced to a minimum capacitance bound, which is analogous to Eq. \eqref{eq:VtrBound} and removes the direct dependence on $p^-$ and $p^+$, by characterizing the ``worst-case'' switching event $p^-\to p^+$
\begin{equation}\label{eq:2BusCap}
    \begin{aligned}
    C >
     &\max_{p^-,\ p^+}
     C_{\mathrm{tr}}\left(p^-,p^+\right)\\
     & \text{subject to}
     & p^- \leq P_{\max}\\
     & & p^+ \leq P_{\max}\\
    \end{aligned}
\end{equation}
The sufficient condition on capacitance is given when both Eqs. \eqref{eq:VtrBound} and \eqref{eq:2BusCap} are satisfied.
This optimization problem can easily be solved computationally to provide a lower bound on the parallel load capacitance.

\section{Generalization to Networks}
In this section we generalize our analysis to arbitrary networks, in the absence of any restrictions on topology. We begin with a few simple assumptions.
First, we assume our network has one strongly-connected component with at least one source.
Second, all sources are assumed to be perfectly regulated with a voltage of $V_0$.
Third, we assume that the aggregate resistance of all lines is bounded from above by $\sum_{(i,j)\in\mathcal{E}} R_{ij}\leq R_{\max}$.
Fourth, we define system loadability, $p_\Sigma$, as the instantaneous sum of constituent loads such that $p_\Sigma = \sum_{k\in\mathcal{L}} p_k$ and assume it is bounded from above by $P_{\max}$ : $p_\Sigma \leq P_{\max} < P_0 = V_0^2/4 R_{\max}$. The maximum system loadability is bounded by $P_0$ of the equivalent two-bus network.
Fifth, we assume a single lower bound across all buses for the equilibrium voltage and voltage during transients given by scalars $V_{\min}$ and $V_\mathrm{tr}$, respectively\footnote{While traditional design practices for known networks may benefit from using non-uniform voltage constraints, the consideration of an arbitrary network without knowledge of of specific loads lends itself toward utilizing consistent bounds across all buses.}.
Finally, the power consumption of each load is bounded from above by $p_k^{\max}$, $p_k \leq p_k^{\max}$, which can vary between loads.

The rest of the section is organized as follows. After brief review of the previously established conditions for existence of equilibrium (Lemma \ref{th:vlow}), we show that for any network the potential $\mathcal{P}$ can serve as a valid Lyapunov function (Corollary \ref{co:LowPbound} and Lemma \ref{CVtr}). We then establish the invariant sublevel set of $\mathcal{P}$ (Theorem \ref{thm:GLowerBound}) and estimate the value of $\mathcal{P}$ after an arbitrary switching event in Lemma \ref{SingleSwitchKE}. Combination of these results allows us to formulate the central condition \eqref{eq:Cbound} on the capacitance guaranteeing the stability of the system.

We start by introducing the high voltage equilibrium voltage profile $v^\mathrm{sep}$ and equilibrium state $x^\mathrm{sep} = [(i^{\mathrm{sep}})^T, (v^\mathrm{sep}_\mathcal{L})^T]^T$ where $i^{\mathrm{sep}} = \mathrm{diag}(R_k^{-1})\nabla v^\mathrm{sep}$. Existence and uniqueness of this equilibrium is guaranteed by the following Lemma \cite{simpson2016voltage}:


\begin{lemma}\label{th:vlow}
Whenever $p_\Sigma < P_0$, there exists exactly one solution to the power flow equations $\partial_v \mathcal{G} = 0$ with all load buses satisfying
    $v_k > V_\mathrm{high}(p_\Sigma, R_{\max})\,$.
At the same time, all other equilibria have at least one load bus $\kappa \in \mathcal{L}$ such that
    $v_\kappa < V_\mathrm{low}(p_\Sigma, R_{\max})\,.$
\end{lemma}
\begin{proof}
See \cite{simpson2016voltage}, Supplementary Theorem 1.
\end{proof}

The quantities $V_\textrm{low}$ and $V_\textrm{high}$ are scalar functions defined in Equations \eqref{eq:Vhigh} and \eqref{eq:Vlow}.
The bounds established in Lemma \ref{th:vlow} demonstrate that exactly one feasible equilibrium point is guaranteed to exist for any $p_\Sigma \leq P_{\max}$ if and only if $R_{\max}$, $P_{\max}$, $V_{\min}$, and $V_0$ satisfy $V_{\mathrm{high}} (P_{\max}, R_{\max}) \geq V_{\min}$ with $V_\mathrm{high}$ given by Eq. \eqref{eq:Vhigh}. This condition is equivalent to
\begin{equation}\label{eq:Pmax_arbitrary}
 P_{\max} \leq \frac{V_{\min}(V_0-V_{\min})}{R_{\max}}\,.
\end{equation}
Both inequalities in Lemma \ref{th:vlow} become tight for the two-bus system---in this sense, the two bus topology (one source separated from multiple loads with total power $p_{\Sigma}$ by a line of resistance $R_{\max}$) is the ``worst-case'' topology for equilibrium point feasibility. This also implies that condition \eqref{eq:Pmax_arbitrary} is both necessary and sufficient for existence of a feasible equilibrium in an ad hoc setting.

\begin{theorem}\label{th:hessian}
The function $\mathcal{G}(v)$ is strictly convex
whenever all load voltages satisfy $v_k > V_0/2$ and $p_\Sigma < P_0$\,.
\end{theorem}
\begin{proof}
The quadratic form of the Hessian is given by
\begin{align}\label{hessform}
 w^T \partial_{vv} \mathcal{G}(v) w = \sum_{(i,j)\in \mathcal{E}}\frac{(w_i-w_j)^2}{R_{ij}} - \sum_{k \in \mathcal{L}} \frac{p_k}{v_k^2}w_k^2,
\end{align}
where we formally define $w_k = 0$ whenever $k \in \mathcal{S}$. For $k\in\mathcal{L}$, let $\Pi_k$ be a path connecting bus $k$ to one of the sources.
Then $w_k = \sum_{(i,j)\in \Pi_k}(w_i-w_j)$. Define $R_{\Pi_k} = \sum_{(i,j)\in \Pi_k} R_{ij}$. The term $w_k^2$ in \eqref{hessform} can be then bounded with the help of Jensen's inequality as
\begin{align}
 w_k^2 &= \left(\sum_{(i,j)\in \Pi_k}(w_i-w_j)\right)^2 \nonumber\\
 &= R_{\Pi_k}^2 \left(\sum_{(i,j)\in \Pi_k}\frac{R_{ij}}{R_{\Pi_k}}\frac{w_i-w_j}{R_{ij}}\right)^2 \nonumber\\
 &\leq R_{\Pi_k}^2  \sum_{(i,j)\in \Pi_k}\frac{R_{ij}}{R_{\Pi_k}}\left(\frac{w_i-w_j}{R_{ij}}\right)^2\nonumber \\
 &\leq R_{\max} \sum_{(i,j) \in \Pi_k} \frac{(w_i-w_j)^2}{R_{ij}}\nonumber
\end{align}
Hence, for the Hessian quadratic form we have
\begin{align}
w^T \partial_{vv} \mathcal{G}(v) w
&> \sum_{(i,j)\in \mathcal{E}}\frac{(w_i-w_j)^2}{R_{ij}} \left(1 - \sum_{k : (i,j)\in \Pi_k} \frac{p_k R_{\max}}{v_k^2}\right)\nonumber \\
 &>\sum_{(i,j)\in \mathcal{E}}\frac{(w_i-w_j)^2}{R_{ij}} \left(1 - \frac{4 p_\Sigma R_{\max}}{V_{0}^2}\right) > 0 \nonumber
\end{align}
\end{proof}
\begin{corollary} \label{co:LowPbound}
The voltage profile $v^\mathrm{sep}$ minimizes the function $\mathcal{G}$ in the domain $v_{k} > V_0/2$ for $k \in \mathcal{L}$. Furthermore in the same domain, and for arbitrary currents, we obtain $\mathcal{P}(x) > \mathcal{G}^\mathrm{sep} = \mathcal{G}(v^\mathrm{sep})$ whenever $x \neq x^{\mathrm{sep}}$ and $\mathcal{P}(x) = \mathcal{G}^\mathrm{sep}$ for $x = x^{\mathrm{sep}}$.
\end{corollary}
Next, we identify the conditions for the decay of the Lyapunov function $\mathcal{P}$ in the transiently acceptable domain of $\mathcal{T} = \{x:\, v_k > V_\mathrm{tr} > V_0/2\}$.
\begin{lemma}
\label{CVtr}
Whenever the capacitances on all the load buses satisfy
\begin{equation} \label{eq:CVtr}
    C_k >\frac{\tau p_k^{\max}}{V_\mathrm{tr}^2},
\end{equation}
one has $\dot{\mathcal{P}} < 0$ whenever $x\in \mathcal{T} \setminus \{x^\mathrm{sep}\}$.
\end{lemma}
\begin{proof}
This result follows directly from the relation $\dot{\mathcal{P}}  = -\dot{x}^T\mathcal{Q}\dot{x}$ and positive definiteness of  the matrix $\mathcal{Q}$ as defined in equation \eqref{eq:Qdef}.
\end{proof}

These two results imply that any sublevel set of $\mathcal{P}$ inside $\mathcal{T}$ is invariant and any trajectory starting inside such a sublevel set converges to the equilibrium point $x^\mathrm{sep}$. These sublevel sets are compact as the function $\mathcal{P}$ is bounded from below by a convex $\mathcal{G}$. To estimate the largest sublevel set of $\mathcal{P}$ that is contained in the transient domain $\mathcal{T}$ we prove the following theorem.
\begin{theorem}
\label{thm:GLowerBound}
The function $\mathcal{P}$ evaluated at the boundary $\partial\mathcal{T}$ is bounded from below such that $\mathcal{P}(x) \geq \mathcal{G}^+_\mathrm{tr}$ where
\begin{equation}\label{eq:vuep_bound}
     \mathcal{G}^+_\mathrm{tr} = \frac{(V_\mathrm{tr} - V_0)^2}{2 R_{\max}} + p_\Sigma^{+} \log V_\mathrm{tr}
\end{equation}
\end{theorem}
\begin{proof}
Given $\mathcal{P} \geq \mathcal{G}$, it is sufficient to establish the bound on $\mathcal{G}$. Given that $x \in\partial\mathcal{T}$, there exists a load bus $\kappa$ with $v_\kappa = V_\mathrm{tr}$. Consider a path $\Pi \subset \mathcal{E}$ connecting the bus $\kappa$ to some source in the system and define $R_\Pi = \sum_{(i,j)\in \Pi} R_{ij}$
Applying the same Jensen's inequality approach as in Theorem \ref{th:hessian}, we show that the potential $\mathcal{G}$ satisfies
\begin{align}
	\mathcal{G}(v) &= \sum_{(i,j)\in \mathcal{E}}\frac{(v_i - v_j)^2}{2 R_{ij}} +\sum_{i\in \mathcal{L}} p_i \log v_i\nonumber \\
    &\geq \sum_{(i,j)\in \Pi}\frac{(v_i - v_j)^2}{2 R_{ij}} + \sum_i p_i \log V_\mathrm{tr} \nonumber\\
    &\geq \frac{(v_\kappa-V_0)^2}{2 R_\Pi} + p_\Sigma \log V_\mathrm{tr} \nonumber\\
    &\geq \frac{(V_\mathrm{tr} - V_0)^2}{2 R_{\max}} + p_\Sigma \log V_\mathrm{tr}\nonumber
\end{align}
\end{proof}

To certify transient stability of the system after the switching event, we estimate the corresponding value of the Lyapunov function $\mathcal{P}$. Specifically we consider an event when only one of the loads $\kappa\in\mathcal{L}$ experiences switching, changing its power from $p_\kappa^-$ to $p_\kappa^+$. 

\begin{lemma}
\label{losses_equiv}
At any equilibrium point, the potential $\mathcal{G}$ can be represented as
\begin{equation}
\mathcal{G}(v)=\sum_{i \in \mathcal{L}} \left[ \frac{p_i}{v_i} \frac{(V_0 - v_i)}{2} + p_i \log(v_i) \right ]
\label{G_losses}
\end{equation}

\end{lemma}
\begin{proof}
Assume that $i_k$ with $k\in\mathcal{V}$ are the nodal currents leaving the sources or the loads. The global current conservation law implies that
\begin{equation}\label{eq:current_conservation}
 \sum_{k \in\mathcal{S}} i_k = - \sum_{k \in\mathcal{L}} i_k
\end{equation}
On the other hand, whenever the voltage on all the source buses is given by $V_0$, and there is no capacitor charging/discharging current at equilibrium, it follows from the energy conservation that the energy produced by the sources is equal to energy consumed by the loads plus the energy dissipated in the lines, or more formally
\begin{equation}\label{eq:power_conservation}
  \sum_{k \in\mathcal{S}} i_k V_0 = - \sum_{k \in\mathcal{L}} i_k v_k + \sum_{(i,j) \in\mathcal{E}} \frac{(v_i - v_j)^2}{R_{ij}}
\end{equation}
Combining the definition \eqref{eq:Gdef} with the relations \eqref{eq:current_conservation} and \eqref{eq:power_conservation} one arrives at \eqref{G_losses}.
\end{proof}

\begin{lemma}
\label{SingleSwitchKE}
For a single load switching event in a network initially at equilibrium, where only one load $\kappa \in\mathcal{L}$ changes, the following bound holds for the potential $\mathcal{P}$:
\begin{equation}
\mathcal{P}(x(t^+)) \leq \mathcal{G}^+_\mathrm{ini} + \frac{\tau_{\mathrm{max}}}{2C_{\kappa}}\left(\frac{p_{\kappa}^--p_{\kappa}^+}{V_{\mathrm{high}}^-} \right)^2 \label{eq:Pplusbound}
\end{equation}
\begin{equation}
\mathcal{G}_{\mathrm{ini}}^+ = \frac{p_\Sigma^-}{2}\frac{V_0-V_{\mathrm{high}}^-}{V_{\mathrm{high}}^-} + p_\Sigma^+ \log V_0
\end{equation}
\end{lemma}

\begin{proof}
We start by bounding the $\mathcal{G}(v(t^+))$. Noting that Lemma \ref{th:vlow} defines the lower bound of voltage level $v_k \geq V_{\mathrm{high}}(p_{\Sigma}^-)$, while the upper bound is $v_k \leq V_0$, we use the expression \eqref{G_losses} in Lemma \ref{losses_equiv} to obtain the following bound:
\begin{align}
    \mathcal{G}^+(v(t^-)) &= \sum_{i \in \mathcal{L}} \left[ \frac{p_i^-}{v_i^-} \frac{(V_0 - v_i^-)}{2} + p_i^+ \log(v_i^-) \right ]\nonumber\\
    & \leq \frac{V_0-V_{\mathrm{high}}^-}{2V_{\mathrm{high}}^-}\sum_{i \in \mathcal{L}} p_i^-+\log(V_0)\sum_{i \in \mathcal{L}}  p_i^+\nonumber
    \\
    &\leq \frac{p_{\Sigma}^-}{2}\frac{V_0-V_{\mathrm{high}}^-}{V_{\mathrm{high}}^-}+p_{\Sigma}^+\log(V_0)   \label{eq:Gplusbound}
\end{align}
Next, due to the continuity of the state variables, and relation \eqref{eq:Current_Dynamics}, one has $L\dot{i}(t^-)=L\dot{i}(t^+)=0$, so only the terms involving $\dot{v}$ change after switching. Given that $p_k \geq 0$ for all $k\in\mathcal{L}$ and $\dot{v}_k =0$ for $k\neq \kappa$ we obtain.
\begin{align}
    \frac{\tau_{\mathrm{max}}}{2}\sum_{k\in\mathcal{L}} C_k\dot{v}^2_k &= \frac{\tau_{\max}}{2} \frac{1}{C_\kappa}\left(-\frac{p_\kappa^+}{v_\kappa^-}-\sum_{\alpha\in \mathcal{E}} \nabla_{\alpha \kappa} i_\alpha\right)^2\nonumber\\
    &\leq\frac{\tau_{\max}}{2C_{\kappa}} \left(\frac{p_\kappa^--p_\kappa^+}{v_\kappa^-}\right)^2\nonumber \\
     &\leq\frac{\tau_{\max}}{2C_{\kappa}} \left(\frac{p_\kappa^--p_\kappa^+}{V_{\mathrm{high}}^-}\right)^2
    \label{eq:single}
\end{align}
Combining the bounds \eqref{eq:single} with \eqref{eq:Gplusbound} we arrive at \eqref{eq:Pplusbound}.
\end{proof}


\begin{theorem}
\label{thm:TransientResult}
The system starting at stable equilibrium and experiencing an arbitrary single-load switching event returns back to a stable equilibrium point whenever the capacitors installed on every load satisfy both Equation \eqref{eq:CVtr} and:

\begin{equation}\label{eq:Cbound}
C_\kappa >
    \underset{p_{\Sigma}^-,\ p_{\Sigma}^+} \max
    \frac{\tau_{\max}}{2\left(\mathcal{G}^+_\mathrm{tr}-\mathcal{G}^+_{\mathrm{ini}}\right)} \left(\frac{p_{\Sigma}^--p_{\Sigma}^+}{V_\mathrm{high}^-}\right)^2
\end{equation}
\begin{equation*}
\begin{aligned}
    \mathrm{subject\,\,to}
    \quad & p^-_\Sigma \leq P_{\max}\\
    & p^+_\Sigma \leq P_{\max}\\
    & |p_{\Sigma}^+-p_{\Sigma}^-|\leq p_{\kappa}^{\max}
\end{aligned}
\end{equation*}
\end{theorem}
\begin{proof}
This result is proven by combining all previous bounds. Specifically, the condition \eqref{eq:Cbound} ensures that the value of the post-switch potential function $\mathcal{P}(x(t^+))$ estimated in \eqref{eq:Pplusbound} is less than $\mathcal{G}^+_\mathrm{tr}$ defined in \eqref{eq:vuep_bound}, which is the minimal value of the potential that any network can achieve within the transiently acceptable domain $\mathcal{T}$. Temporal decay of $\mathcal{P}$ established by Lemma \ref{CVtr} under assumption \eqref{eq:CVtr} implies that the system will stay inside $\mathcal{T}$. Therefore, in accordance to Lyapunov theorem, positive-definiteness of $\mathcal{P}(x)- \mathcal{G}^{\mathrm{sep}}$ (Corollary \ref{co:LowPbound}) and its temporal decay (Lemma \ref{CVtr}) inside $\mathcal{T}\setminus \{x^\mathrm{sep}\}$ implies that the system will converge to $x^\mathrm{sep}$.

\end{proof}
\begin{remark}\label{remark:Pcrit}
For each $V_\mathrm{tr}$, the sufficient bound for $C$ defined by \eqref{eq:Cbound} exists only for small enough values of power, $p_\kappa^{\max} \leq P_{\kappa}^{\mathrm{crit}}$. Above those levels, stability cannot be certified for any capacitance as the initial energy $\mathcal{G}_{\mathrm{ini}}^+$ may exceed $\mathcal{G}_{\mathrm{tr}}^+$ for some admissible values of $p_\Sigma^-, p_\Sigma^+$.
\end{remark}
\begin{remark}
Numerical simulation demonstrates that the worst case switching scenario (the scenario that maximizes $C_{\kappa,\mathrm{tr}}$) corresponds to $p_\Sigma^- =  P_{\max} - p_{\kappa}^{\max}$ and $p_\Sigma^+ = P_{\max}$, that is load $\kappa$ switching on to its maximum power and bringing the total network loading to $P_{\max}$. 
\end{remark}
The nature of the derivation above implies that the condition \eqref{eq:Cbound} is sufficient but not necessary. The following Lemma introduces a necessary condition:
\begin{lemma}
For a system to maintain stability in an ad hoc setting, it is necessary that each load capacitance satisfies
\begin{equation}
\label{eq:Css}
    C_{\kappa} > \frac{\tau p_k^{\max}}{V_{\min}^2}
\end{equation}
\end{lemma}
\begin{proof}
    As follows from the discussion in section \ref{sec:twobus}, violation of this condition results in the loss of asymptotic stability for a two-bus system with load power $p^{\max} = P_{\max}$.
\end{proof}
Tighter necessary conditions (not shown) can be obtained numerically by simulating transients in specific networks, and these conditions are in close agreement with \eqref{eq:Css}. 

\section{Discussion}

A comparison of the lower bounds on $C$, normalized by $C_0=\tau_{\max}/R_{\max}$, during a switching event, whose magnitude is normalized by the natural unit of power  $P_0 = V_0^2/4 R_{\max}$, is presented in Figure \ref{fig:compiled_results}. The larger the load, the larger the capacitance required to ensure stability. The small signal stability constraint (Eq. \eqref{eq:Css}) is a necessary condition while the transient stability constraint ( Eqs. \eqref{eq:CVtr} and \eqref{eq:Cbound}) is a sufficient condition. The gap between these constraints gives an indication of how conservative the sufficient bounds are.


The choice of $V_\mathrm{tr}$ alters the capacitance constraints. Increasing the levels of $V_{\mathrm{tr}}$ decreases the requirements imposed by Eq. \eqref{eq:CVtr} but increases the ones from \eqref{eq:Cbound}. Furthermore, high levels of $V_{\mathrm{tr}}$ result in relatively small critical levels of load power as discussed in Remark \ref{remark:Pcrit}. For example, given $V_\mathrm{tr}=0.66V_0$ as in Figure \ref{fig:compiled_results}, the maximum load power consumption is $P_{\kappa}^{\mathrm{crit}} \approx 0.47P_0$. 
A trade-off therefore exists between power demanded and magnitude of the capacitance as well as between choice of $V_\mathrm{tr}$ and $P_k^{\mathrm{crit}}$.

\begin{figure}
\centering
\includegraphics[width =0.9\linewidth]{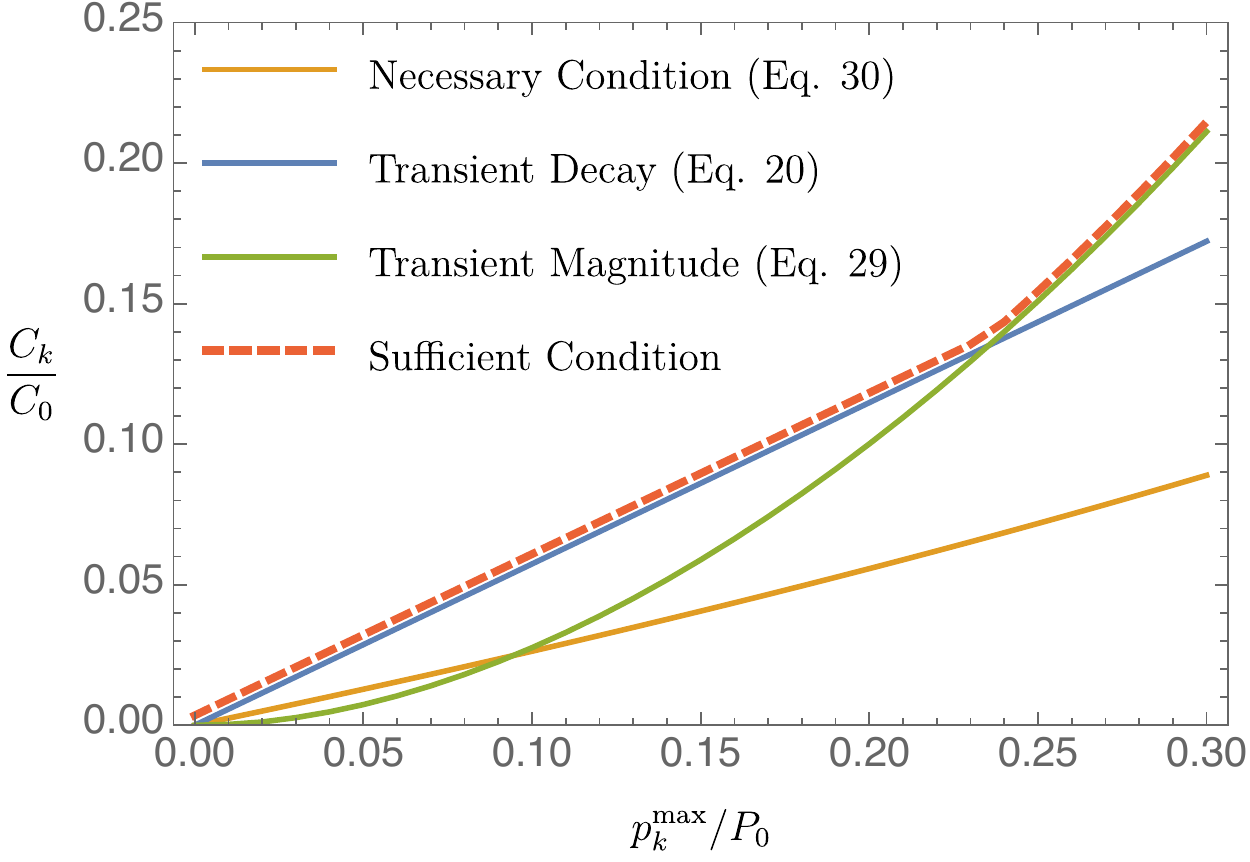}
\caption{Comparison of lower bounds on $C_k/C_0$ for $V_{\mathrm{tr}}=0.66V_0$.}
\label{fig:compiled_results}
\end{figure}

\section{Proposed Design Scheme And Path Forward}

We have established a lower bound on load capacitance which can guarantee network stability without pre-determining the network topology. Our bound provides a theoretical foundation for ad hoc microgrids with modular source and load units. The development process for these microgrids is
\begin{enumerate}
\item{Define acceptable voltage levels across all units based on converter constraints: nominal voltage $V_0$, minimal acceptable equilibrium voltage $V_\text{min}$, minimum acceptable voltage during switching transients $V_\mathrm{tr}$.}
\item{Select system parameters: the upper bound on system loading $P_{\max}$ and the maximum line resistance $R_{\max}$ (determined by line material, diameter and length).}
\item{For each load $\kappa$ with maximum power $p_{\kappa}^{\max}$, select a capacitance that ensures stability according to the constraints in Figure \ref{fig:compiled_results}.}
\end{enumerate}



This process is independent of the network topology and therefore it does not need to be repeated for each community. Instead, it can be performed once to develop, for example, electricity access ``kits'' that could be produced in bulk and easily adapted to the changing needs of individual communities without oversight.



Several exciting paths needs to be further explored. These include the generalization of the results to restricted topologies, more detailed load and source models. Similarly, more research is required to understand how stability can be enforced on secondary control loops on sources \cite{zhao2015distributed, cdc2016,  de2016power} in the presence of inductive delays. 

\bibliographystyle{ieeetr}
\bibliography{biblio.bbl}

\end{document}